\documentclass{article}

\usepackage{arxiv}

\usepackage[utf8]{inputenc} 
\usepackage[T1]{fontenc}    
\usepackage{hyperref}       
\usepackage{url}            
\usepackage{booktabs}       
\usepackage{amsfonts}       
\usepackage{nicefrac}       
\usepackage{microtype}      
\usepackage{lipsum}		
\usepackage{graphicx}
\usepackage{natbib}
\usepackage{doi}
\usepackage{multicol}

\newcommand{\real}{\mathbb{R}}

\renewcommand{\real}{\mathbb{R}}

\newcommand{\lagr}{\mathcal{F}}
\newcommand{\En}{\operatorname{E}}
\newcommand{\xh}{x_{(h)}}
\newcommand{\ssubset}{\subset\joinrel\subset}

\DeclareSymbolFont{bbold}{U}{bbold}{m}{n}
\DeclareSymbolFontAlphabet{\mathbbold}{bbold}

\newcommand{\vect}[1]{\mathbbold{#1}}
\newcommand{\vectorones}[1][]{\vect{1}_{#1}}
\newcommand{\vectorzeros}[1][]{\vect{0}_{#1}}

\usepackage{amsmath}
\usepackage{amssymb}
\usepackage{amsthm}
\usepackage{xcolor}
\usepackage{tikz}
\usetikzlibrary{matrix, backgrounds, calc}
\usepackage{hyperref}       
\usepackage{subfig}
\usepackage{mathtools, nccmath}
\usepackage{upgreek}
\usepackage{cancel}

\usepackage{graphicx}

\newtheorem{theorem}{Theorem}

\newtheorem{corollary}[theorem]{Corollary}
\newtheorem{proposition}[theorem]{Proposition}
\newtheorem{definition}[theorem]{Definition}

\newtheorem{remark}{Remark}
\newtheorem{problem}{Problem}

\definecolor{gnblue4}{RGB}{0,108,212} 
\definecolor{gnblue6}{RGB}{35,156,255}
\definecolor{gnpurple4}{RGB}{72,46,146} 
\hypersetup{colorlinks=true, linkcolor=gnpurple4, breaklinks=true, urlcolor=gnpurple4, citecolor=gnpurple4}

\title{Hybrid Energy-Based Models for Physical AI: Provably Stable Identification of Port-Hamiltonian Dynamics}

\author{ \href{https://orcid.org/0009-0000-3444-0838}{\includegraphics[scale=0.06]{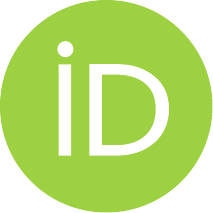}\hspace{1mm}Simone Betteti} \\
RIAS Lab\\ 
	The Italian Institute of\\
	Artificial Intelligence for Industry\\
	Turin, 10129, IT \\
	\texttt{simone.betteti[at]ai4i.it} \\
	\And
	\href{https://orcid.org/0000-0003-1190-6097}{\includegraphics[scale=0.06]{Fig/orcid.pdf}\hspace{1mm}Luca Laurenti\thanks{L.L. is also with the Delft Center for Systems and Control, TU Delft, Delft, 2600 AA, The Netherlands.}} \\
RIAS Lab\\ 
	The Italian Institute of\\
	Artificial Intelligence for Industry\\
	Turin, 10129, IT \\
	\texttt{luca.laurenti[at]ai4i.it} \\
}

\renewcommand{\shorttitle}{Hybrid Energy-Based Models for Physical AI}

\hypersetup{
pdftitle={Hybrid Energy-Based Models for Physical AI: Provably Stable Identification of Port-Hamiltonian Dynamics},
pdfsubject={eess.SY, cs.LG, cs.AI},
pdfauthor={Simone Betteti, Luca Laurenti},
pdfkeywords={Energy-based models, Physical AI, Stable System Identification},
}

\begin{document}
\maketitle

\begin{abstract}

Energy-based models (EBMs) implement inference as gradient descent on a learned Lyapunov function, yielding interpretable, structure-preserving alternatives to black-box neural ODEs and aligning naturally with physical AI. Yet their use in system identification remains limited, and existing architectures lack formal stability guarantees that globally preclude unstable modes. We address this gap by introducing an EBM framework for system identification with stable, dissipative, absorbing invariant dynamics. Unlike classical global Lyapunov stability, absorbing invariance expands the class of stability-preserving architectures, enabling more flexible and expressive EBMs. We extend EBM theory to nonsmooth activations by establishing negative energy dissipation via Clarke derivatives and deriving new conditions for radial unboundedness, exposing a stability-expressivity tradeoff in standard EBMs. To overcome this, we introduce a hybrid architecture with a dynamical visible layer and static hidden layers, prove absorbing invariance under mild assumptions, and show that these guarantees extend to port-Hamiltonian EBMs. Experiments on metric-deformed multi-well and ring systems validate the approach, showcasing how our hybrid EBM architecture combines expressivity with sound and provable safety guarantees by design.

\end{abstract}

\section{INTRODUCTION}
\emph{Energy-based models} (EBMs) are a class of machine-learning architectures in which neural dynamics arise as the gradient flow of a Hopfield-type Lyapunov function. Rooted in the pioneering work of~\cite{CMA-GS:83} and~\cite{HJJ:82, HJJ:84}, EBMs encode meaningful data patterns as stable equilibria of a dissipative dynamical system, ensuring convergence towards stored representations through gradient descent on an energy landscape. However, the classical Hopfield construction offered limited expressive capacity for modern learning tasks, with the number of retrievable patterns scaling as $N/\log(N)$ for networks of size $N$~\citep{MR-PE-RE:87}. Recent generalizations of Hopfield networks~\citep{KD-HJJ:20, HB-LY-PB:23} have dramatically expanded their applicability. Modern EBMs preserve the foundational energy-based structure while exhibiting formal parallels with Transformer architectures~\citep{RH-SB-LJ:21}, enabling an exponential scaling of capacity in $N$~\citep{DM-HJ-LM:17}. Their alignment with contemporary neural architectures has facilitated successful deployments across standard machine-learning domains such as image classification~\citep{HB-CDH-SH:24}, image generation~\citep{AL:24, PGY-KJ-KB:23}, and in-context learning~\citep{WW-HTY-HJYC:25}. Canonically, learning in \emph{energy-based models} relies on optimizing the network parameters to map relevant data-class centroids to stable equilibria and their associated basin of attraction in the \emph{energy} landscape.  

Despite their dynamical foundation, EBMs have rarely been deployed for system identification, a domain where the underlying objective of capturing the evolution of nonlinear systems from data is inherently dynamical. System identification~\citep{AKJ-EP:71, LL:10} sits at the interface of control theory and optimization and has progressively integrated machine-learning components~\citep{CRTQ-RY-BJ:18} to accommodate increasingly complex real-world systems. This emerging direction, now widely referred to as \emph{physical AI}, has driven notable progress in robotics~\citep{PA-GR-YW:22}, power-grid modeling~\citep{CR-JH-SGS:22}, and industrial manufacturing~\citep{DP-DC-HJA:18}. Yet, the safety-critical nature of these domains demands not only expressive models, but also structural guarantees of stability and robustness~\citep{DH-LB-YL:21}. Recent efforts with dissipative neural networks and stable neural ODEs address this need by enforcing parameter constraints~\citep{DJ-TA-VS:22, XY-SS:23}. Notably, recent work on learning Port-Hamiltonian dynamics~\citep{MS-PM-BM:20,RFJ-KDK-KM:25} has centered on learning optimality and empirical benchmarking, with limited attention to stability guarantees or architectural considerations. In summary, existing neural identifiers either lack formal stability guarantees outside the energy and dissipative frameworks or require strong smoothness and convexity assumptions.

\noindent\textbf{Contributions.}
This work introduces an energy-based system identification framework that reconciles stability with expressive modeling. Its main contributions are:
\begin{itemize}
    \item a generalized Lyapunov stability analysis for EBMs with nonsmooth activations, establishing energy dissipation via Clarke derivatives and revealing a structural stability-expressivity trade-off;
    \item a \emph{hybrid EBM architecture} with dissipative, absorbing invariant visible-layer dynamics, guaranteeing bounded trajectories while avoiding the restrictive conditions imposed by fully recurrent EBMs, and static hidden layer maps for fast and efficient inference;
    \item a Port-Hamiltonian extension enabling identification of nonlinear dynamics under state-dependent metrics and rotational perturbations, validated on metric-deformed multi-well and ring systems.
\end{itemize}
Section~II presents the necessary preliminaries. Section~III introduces EBMs and formulates the problem. Section~IV extends classical $C^2$ dissipation results to locally Lipschitz activations and establishes new conditions for radial unboundedness. Section~V develops the hybrid architecture and proves absorbing invariance, including its Port-Hamiltonian generalization. Section~VI reports numerical validation, and Section~VII concludes with future directions.

\section{PRELIMINARIES}
\noindent \textbf{Notation.} 
$\vectorones[d]$ denotes the $d$-dimensional vector with all ones, $\vectorzeros[d]$ the $d$-dimensional vector with all zeros, while $\mathcal{I}_d$ denotes the identity matrix in $\real^{d\times d}$. We denote compact subsets of $\real^d$ as sets $\mathcal{B}\ssubset\real^d$, and their boundary as $\partial\mathcal{B}$. For two real vectors $x,y$ of the same dimension, $x^\top y$ denotes the standard inner product. Let $f:\real^d\to\real$; for $f\in C^k(\real^d)$, the function is $k$-times continuously differentiable. For $f\in C^2(\real^d)$, the gradient of $f$ is denoted as $\nabla f$, and the Hessian as $D(\nabla f)=D^2f$. The partial derivative of $f$ with respect to the variable $x_{i}$ is denoted as  $\partial f/\partial{x_{i}}$. $\Theta$ bounds the growth of $f\sim\Theta(g)$ through the existence of constants $a,b>0$ such and a function $g$ such that $a\ g(x)< f(x) < b\ g(x)$. Given functions $g:\real^n\to\real^d$ and $f:\real^d\to\real$, we denote the composition of the two functions as $f\circ g(y)$, for $y\in\real^n$. The abbreviation a.e. stands for almost everywhere, that is for all $\mathcal{B}\subset \real^d$ except zero measure sets. Given a matrix $A\in\real^{d\times d}$, we denote with $A^\top$ its transpose. In case $A$ is symmetric, $\uplambda_{\min}(A),\ \uplambda_{\max}(A)\in\real$ denote its minimum and its maximum eigenvalue. We denote with $B_r(x)\subset \real^d$ the ball of radius $r>0$ and center $x\in\real^d$.

\subsection{Lie and Clarke's derivatives}
\begin{definition}[Lie derivative]\label{def: lie}
    Let $X:\real^{N}\to\real^{N}$ be a locally Lipschitz continuous vector field and let $\Phi_{X}:\real_{\geq 0}\to\real^{N}$ denote the associated local flow. For a smooth function $g:\real^{N}\to\real$, the Lie derivative of $g$ along $X$ is defined as
    \begin{equation}\label{eq: lie-der}
        \mathcal{L}_{X}g(x) = \frac{d}{dt}|_{t=0} (g\ \circ \Phi_{X}^{t} )(x),
    \end{equation}
    and coincides with the directional derivative
    \begin{equation}\label{eq: dir-der}
        \mathcal{L}_{X}g(x)=\nabla_{X}g(x)=\nabla_{x}g(x)^\top X(x).
    \end{equation}
\end{definition}
Lie derivatives~\citep[Chapters 3,~9]{LJM:12} will be central for the characterization of the negative definiteness of the \emph{energy} function along the trajectories generated by the EBM dynamics. We extend the notion of derivatives to nonsmooth functions by recalling Clarke's generalized gradient and directional derivative~\citep{CFH:75}, which apply to any locally Lipschitz function.
\begin{definition}[Generalized gradient]\label{eq: gen-grad}
Let $f:\real^N\to\real$ be a locally Lipschitz function. The generalized gradient of $f$ at $x\in\real^N$ is denoted as $\partial f(x)$ and is the convex hull of the set of limits
\begin{equation}
	\lim_{k\to +\infty} \nabla f(x+h_k),
\end{equation}
with $h_k\xrightarrow[]{k\to +\infty} \vectorzeros[N]$ and $f$ differentiable at $x+h_k\in\real^N$ for all $k\in\mathbb{N}$.
\end{definition}
\begin{definition}[Generalized directional derivative]\label{def: gen-dir}
Let $f:\real^N\to\real$ be a locally Lipschitz function and take $v\in\real^N$. Then the Clarke's generalized directional derivative at $x\in\real^N$ is defined as
\begin{align}
	f^{\circ}(x;v)&=\lim_{h\to\vectorzeros[N]}\sup_{\updelta\to 0} \frac{f(x+h+\updelta v)-f(x+h)}{\updelta}\\
	&=\max_{\upxi\in\partial f(x)} \upxi^\top v\label{eq: gen-dir}.
\end{align}
\end{definition}
In particular, combining~\eqref{eq: dir-der} with~\eqref{eq: gen-dir}, the Clarke Lie derivative 
of a locally Lipschitz function $f$ along a vector $v$ is given by
\begin{equation}\label{eq: gen-Lie}
    \mathcal{L}_v f(x)
        := f^{\circ}(x;v)
        = \max_{\xi\in\partial f(x)} \xi^\top v,
\end{equation}
extending the classical smooth definition to the nonsmooth setting.

\subsection{Conjugate (Legendre) transform}
\begin{definition}[Conjugate transform]
    Let $\lagr:\real^N\to\real$ be proper and convex. The \emph{conjugate transform} of $\mathcal{F}$ (also referred to as \emph{Legendre transform}) is defined as
    \begin{equation}\label{eq: conj-trans}
        \lagr^{\star}(y) = \sup_{x\in\real^N} x^\top y - \lagr(x).
    \end{equation}
\end{definition}
Standard results from convex analysis~\citep[Chapter 13]{BHH-CPL:17} provide information on the convexity and asymptotic growth of the conjugate transform. For polynomially growing convex functions $\lagr(x)\sim\Theta(\|x\|^q)$, the conjugate $\lagr^\star$ is also convex and $\lagr^\star(y)\sim\Theta(\|y\|^\frac{q}{q-1})$.
\begin{remark}[Smooth proper convex functions]\label{rm: smooth}
    If $\lagr\in C^{1}(\real^N)$, then any maximizer of \eqref{eq: conj-trans} satisfies the first-order optimality condition
    \begin{align}
        \vectorzeros[N] &=\nabla_{x}[x^\top y-\lagr(x)]\nonumber\\
        &= y-\nabla\lagr(x).
    \end{align}
    It then follows that $\lagr^\star(y)_{|y=\nabla\lagr(x)}=x^{\top}\nabla\lagr(x)-\lagr(x)$. Combining this with the growth of the conjugate $\lagr^\star$ yields
    \begin{align}\label{eq: sm-growth}
        \lagr^\star(y)_{y=\nabla\lagr(x)}&\sim \Theta(\|y\|^{q/(q-1)}_{|y=\nabla\lagr(x)})\nonumber\\
        &\sim \Theta(\|x\|^{(q-1)\frac{q}{q-1}}) = \Theta(\|x\|^{q}),
    \end{align}
    where the final equivalence uses the fact that the gradient of a $q$-coercive convex function grows $(q-1)$-polynomially in $\|x\|$.
\end{remark}

\section{PROBLEM STATEMENT}
EBMs were initially introduced in~\citep{KD-HJJ:16,KD-HJJ:20} as fully recurrent architectures, where all units interact with one another. Subsequent work specialized EBMs into modular, layered architectures~\citep{HB-CDH-SH:22,KL-SJJ-KD:25} to improve learning efficiency and compatibility with standard MLP-based frameworks. In this paper, we consider EBMs as defined by the autonomous nonlinear dynamical systems in Definition~\ref{def: EBMs} below, with parametric interactions between units, and that can be trained for machine-learning tasks using standard optimization techniques~\citep{BL-CFE-NJ:18}.
\begin{definition}[Energy-based models (EBMs) ]\label{def: EBMs}
    Let $\lagr:\real^{N}\to\real$ be a proper, convex, and $C^1$ function, and define $\Psi(x)=\nabla \lagr(x)$ to be the activation function of the \emph{energy-based model}. Further, let $W\in\real^{N\times N}$ be the symmetric matrix of parameters and $b\in\real^N$ a bias vector. Then, an \emph{energy-based model} is defined as a dynamical system 
    \begin{equation}\label{eq: EBM}
        \dot x = -x +W\Psi(x)+b,
    \end{equation}
    and has associated \emph{energy}
    \begin{equation}\label{eq: EBM-ene}
        \En(x)=-\frac{1}{2}\Psi(x)^{\top}W\Psi(x)+(x-b)^{\top}\Psi(x)-\lagr(x).
    \end{equation}
\end{definition}
In classical EBM theory, stability of the dynamics~\eqref{eq: EBM} is ensured by proving that the energy function $\En(x)$ in~\eqref{eq: EBM-ene} is dissipative along the dynamics~\eqref{eq: EBM} and radially unbounded. As shown in Section~IV, this Lyapunov-based viewpoint imposes restrictive growth conditions on the activation functions, leading to a \emph{stability-expressivity trade-off}. In contrast, this paper relies on a different notion of stability for system~\ref{def: EBMs}, which leverages absorbing invariant sets and, as we will show, allows for more expressive \emph{stable-by-design} architectures. An absorbing invariant set, in the sense of Definition~\ref{def: abs-inv}, is a region that all trajectories are eventually drawn into (absorbing) and from which none can escape once inside (invariant).
\begin{definition}{(Absorbing invariant set)}\label{def: abs-inv}
Let $\mathcal{D}\ssubset \real^N$ be a compact set and $X:\real^N\to\real^N$ a locally Lipschitz vector field. Consider $\Phi_X^t$ the flow induced by $X$. Then the set $\mathcal{D}$ is
\begin{itemize}
    \item[(i)] absorbing if $\lim_{t\to +\infty}\Phi_X^t(x) \in\mathcal{D}$ for all $x\in\real^N$;
    \item[(ii)] invariant if $X(x)^\top \upeta(x)\leq 0$ for all $x\in\partial\mathcal{D}$, with $\upeta(x)\in\real^N$ outer normal to $\partial\mathcal{D}$ at $x$.
\end{itemize}
\end{definition}
An absorbing invariant set as in Definition~\ref{def: abs-inv} provides a natural \emph{safe set} for the EBM dynamics~\eqref{eq: EBM}, confining trajectories to a predictable region of the state space. 
Dissipativity further enforces non-increasing energy, ruling out divergence and steering trajectories toward energetically stable subsets of $\mathcal{D}$, including any equilibria contained within it. Neural models that guarantee such bounded and predictable behavior are essential for deploying learning-based identifiers in safety-critical settings. Formally, in this paper we consider the following problem:
\begin{problem}
\label{prob:ProbForm}
Given samples $(x,f^\star(x)) \in \mathcal{X}$ from an unknown autonomous system $\dot{x} = f^\star(x)$, our objective is to identify the underlying dynamics through a EBM~\eqref{eq: EBM} $\dot{x} = f(x;W)$ for which there exists an invariant absorbing compact set $\mathcal{D}$ containing the equilibrium points of $f^\star$.
\end{problem}
In the following sections we show that Problem~\ref{prob:ProbForm} can be addressed by relaxing the assumptions typically required to guarantee Lyapunov asymptotic stability for EBMs~\citep{KD-HJJ:16}. Our approach is non-trivial: it removes the stringent conditions imposed by classical Lyapunov arguments while still yielding a \emph{stable-by-design} neural architecture suitable for reliable system identification. Existing EBM formulations either remain bounded to Lyapunov-based constraints or rely on strongly convex parameterizations, offering far less flexibility than the framework proposed here. In Section~IV, we extend standard dissipativity arguments to nonsmooth activations and show that Lyapunov asymptotic stability enforces a sublinearity condition that restricts expressivity. Section~V introduces the hybrid EBM architecture, which circumvents this limitation by requiring bounded activation only in the first hidden layer. This boundedness confines the contributions of all deeper hidden layers to compact sets, allowing us to construct a conservative absorbing invariant set $\mathcal{D}$ for the hybrid dynamics.

\section{LYAPUNOV STABILITY AND THE STABILITY-EXPRESSIVITY TRADE-OFF}
Classical EBM analyses assume $\lagr\in C^2(\mathbb{R}^N)$ so that the Hessian $D^2\lagr$ is well defined and positive semidefinite, ensuring energy dissipation. However, practical architectures often use activations that are only locally Lipschitz and thus differentiable almost everywhere. We therefore extend the standard $C^2$ argument by employing Clarke generalized derivatives~\citep{CFH:75}, which recover the same negative-definiteness property of the Lie derivative without requiring twice differentiability. For simplicity we take $b=\vectorzeros[N]$, though the results extend to the general case. We begin by establishing dissipativity for nonsmooth activations in Proposition~\ref{prop: dis}.
\begin{proposition}[Negative definiteness of the energy derivative]\label{prop: dis}
    Let $\lagr\in C^1(\real^N)$ and let $f_{E}(x)=-x+W\Psi(x)$ be the vector field associated to the EBM dynamics~\eqref{eq: EBM}. Then for all $x\in\real^N$ and $\En(x)$ defined in~\eqref{eq: EBM-ene}
    \begin{equation}
        \mathcal{L}_{f_E}\En(x)\leq 0\qquad \text{a.e.}
    \end{equation}
\end{proposition}
\begin{proof}
    Since $\Psi=\nabla\lagr$ is locally Lipschitz, it is differentiable a.e., and at such points
    \begin{align}
        \nabla \En(x)=&-D\Psi(x)W\Psi(x)+D\Psi(x)x\nonumber\\
                    =&-D\Psi(x)f_{E}(x),
    \end{align}
    where in the first passage we have exploited the symmetry of $D\Psi(x)$ and $W$. By Rademacher theorem~\citep[Chapter 6]{ELC-GRF:15}, $D\Psi(x)$ is defined a.e., and by the convexity of $\lagr$ every element $\Sigma(x)\in\partial(D\Psi)(x)$ of the Clarke generalized Hessian is positive semi-definite, i.e. $\Sigma(x)\succeq 0$ for all $x\in\real^N$. The Clarke generalized gradient of the \emph{energy} is then
    \begin{equation}
        \partial\En(x)=\{-\Sigma(x)f_{E}(x):\: \Sigma(x)\in\partial(D\Psi)(x)\}.
    \end{equation}
    The Clark Lie derivative is then
    \begin{align}
        \mathcal{L}_{f_E}\En(x)&=\max_{v\in\partial\En(x)} v^\top f_{E}(x)\nonumber\\
        &=\max_{\Sigma(x)\in\partial(D\Psi)(x)} -f_{E}(x)^\top\Sigma(x)f_{E}(x)\nonumber\\
        &= \max_{\Sigma(x)\in\partial(D\Psi)(x)} -\|f_{E}(x)\|_{\Sigma(x)^{1/2}}^2\leq 0.
    \end{align}
\end{proof}
The classical EBM literature focuses solely on dissipation and convergence, overlooking the crucial issue of radial unboundedness of the energy. Proposition~\ref{prop: EBM-ru} below provides necessary and sufficient conditions for radial unboundedness in a fully connected EBM architecture, without which trajectories may escape to infinity, producing pathological behaviors such as diverging activity or exploding gradients during training. We consider the realistic case in which $W\in\real^{N\times N}$, learned from data-driven optimization, is not necessarily negative definite and has positive eigenvalues driving the quadratic term of the energy to $-\infty$.
\begin{proposition}[Radial unboudedness of the energy derivative]\label{prop: EBM-ru}
    Let $W\in\real^{N\times N}$ with no zero entries and such that $\uplambda_{\max}(W)>0$. Let $\lagr:\real^N\to\real$ be a proper, convex and $C^1$ function. Suppose that $\lagr(x)\sim \Theta(\|x\|^q)$ with $q\in]1,+\infty[$. Then
    \begin{equation}\label{eq: EBM-ene-ru}
        \liminf_{\|x\|\to +\infty} \En(x) = +\infty\qquad\iff\qquad q<2.
    \end{equation}
\end{proposition}
\begin{proof}
    Let $q\in]1,+\infty[$ and consider the decomposition of~\eqref{eq: EBM-ene} as $\En(x)=\En_{Quad}(x)+\En_{Ham}(x)$.
    \begin{itemize}
        \item[(i)] $\En_{Quad}(x)=-\frac{1}{2}\Psi(x)^\top W\Psi(x)$: since $\lagr(x)\sim \Theta(\|x\|^q)$, then we have that $\Psi(x)=\nabla\lagr(x)\sim \Theta(\|x\|^{q-1})$. Since $W$ has no zero entries, the greatest growing entries of the quadratic product grow as $|z|^{2(q-1)}$ for $z\in\real$, and in fact we have
            \begin{align}
                |\nabla\lagr(x)^\top W\nabla\lagr(x)|&\leq \uplambda_{\max}(W)\ \|\nabla\lagr(x)\|^{2}_2\nonumber.
            \end{align}
            Therefore, $-\En_{Quad}(x)\sim \Theta(\|x\|^{2(q-1)})$.
        \item[(ii)] $\En_{Ham}(x)=x^\top\Psi(x)-\lagr(x)$: from the convex duality result (Remark~\ref{rm: smooth}), $\En_{Ham}(x)\sim \Theta(\|x\|^q)$.
    \end{itemize}
    Consequently, the \emph{energy} $\En(x)=\En_{Quad}(x)+\En_{Ham}(x)$ is radially unbounded if
    \begin{equation}
        \liminf_{\|x\|\to +\infty} -\|\nabla\lagr(x)\|^{2}_2+x^{\top}\nabla\lagr(x)-\lagr(x)=+\infty.
    \end{equation}
    This is equivalent to require that the growth $q$ of $\En_{Ham}$ dominates the growth $2(q-1)$ of $\En_{Quad}$, which happens iff
    \begin{equation}
        q>2(q-1) \iff q<2.
    \end{equation}
\end{proof}
Proposition~\ref{prop: EBM-ru} implies that fully recurrent EBMs are fundamentally limited: stability requires sublinear activations $q\in(1,2)$, which does not make possible the use of high-degree polynomial activations advocated in~\citep{KD-HJJ:16,KD-HJJ:20}. Even multilayer feedforward EBMs~\citep{HB-CDH-SH:22} only soften this restriction, still requiring sublinear activations in every second layer to maintain radial unboundedness. Practical models bypass these constraints with heavy normalization throughout the network. In contrast, the hybrid architecture introduced next achieves stability via invariance with normalization applied solely to the first hidden layer, fully restoring expressivity while guaranteeing bounded trajectories.

\section{MULTI-LAYER EBMs AND ABSORBING INVARIANCE}
The main contribution of our work is the introduction of a hybrid architecture that preserves both computational expressivity and stability of the EBM by requiring boundedness of the activation only in the first hidden layer, and allowing for arbitrary activations in all the deeper layers. We start by restricting our attention to layered architectures, a standard choice in applications, where each hidden layer interacts only with its immediate neighbors.
\begin{definition}[Multi-layer EBMs]\label{def: m-EBMs}
    Let $L\in\mathbb{N}$ be the number of layers, and denote $N_{h}\in\mathbb{N}$ the width of the $h$-layer, with $N=\sum_{h=1}^L N_h$. For each layer, let $\xh\in\real^{N_{h}}$ be its state, and define a family $\{\lagr_{h}\}_{h=1}^L$ of proper, convex, and $C^1$ functions such that
    \begin{equation}\label{eq: m-Lagr}
        \lagr(x) = \sum_{h=1}^L \lagr_h(\xh).
    \end{equation}
    Define the interaction matrix $W\in\real^{N\times N}$ as the block tridiagonal matrix
    \begin{align}\label{eq: m-W}
        W=
        \begin{pmatrix}
            \vectorzeros[N_1\times N_1] & W_{12} & \vectorzeros[N_1\times N_3] & \cdots & \vectorzeros[N_1\times N_L] \\
            W_{12}^\top & \vectorzeros[N_2\times N_2] & W_{23} & \ddots & \vectorzeros[N_2\times N_L] \\
            \vectorzeros[N_2 \times N_1] & W_{23}^\top & \vectorzeros[N_3 \times N_3] & \ddots & \vdots \\
            \vdots & \cdots & \ddots & \ddots & W_{(L-1)L} \\
            \vectorzeros[N_L\times N_1] & \cdots & \cdots & W_{(L-1)L}^\top & \vectorzeros[N_L\times N_L]
        \end{pmatrix}.
    \end{align}
    Then the dynamics of each layer of the EBM architecture are
    \begin{align}
        \dot{x}_{(1)}=&-x_{(1)}+W_{12}\Psi_2(x_{(2)}),\label{eq: l-EBM-1}\\
        \dot{x}_{(h)}=&-\xh+W_{h(h-1)}\Psi_{h-1}(x_{(h-1)})\nonumber\\&+W_{h(h+1)}\Psi_{h+1}(x_{(h+1)}),\quad h=2,\dots,L-1\label{eq: l-EBM-h}\\
        \dot{x}_{(L)}=&-x_{(L)}+W_{L(L-1)}\Psi_{L-1}(x_{(L-1)})\label{eq: l-EBM-L}.
    \end{align}
\end{definition}
For system-identification tasks based on $0$-step encoding-decoding maps $x \mapsto f(x)$, only the visible layer $x_{(1)}$ from~\ref{eq: l-EBM-1} represents measured or reconstructed physical variables. A full multilayer EBM, however, evolves in a higher-dimensional state and converges to an equilibrium determined by all layers, creating two practical difficulties: (i) \textbf{temporal alignment}, i.e., determining when the visible state should be matched with data; and (ii) \textbf{backpropagation through depth}, since hidden layers introduce recurrent dependencies that complicate gradient computation. To preserve the dynamical structure of the visible layer while removing internal recurrences, we adopt a hybrid model in which hidden-layer activity is computed through feedforward maps, whereas the visible layer retains its full EBM dynamics.
\begin{definition}[Hybrid EBM architecure]\label{def: h-EBM}
    Let $\{\upphi^f_{h}\}_{h=2}^L$ be a family of continuous functions, $\upphi^f_h:\real^{N_{h-1}}\to\real^{N_h}$, and define the state of each hidden layer of the multi-layer EBM in Definition~\ref{def: m-EBMs} via the feedforward map
    \begin{equation}\label{eq: feedf}
        x_{(h-1)}\mapsto \upphi_{h}^f(x_{(h-1)})=\xh.
    \end{equation}
    The projected visible-layer dynamics in~\eqref{eq: l-EBM-1} are then defined as the projected energy gradient:
    \begin{align}\label{eq: visible-d}
        \dot{x}_{(1)}&=-\nabla_{x_{(1)}} \En(x_{(1)},\upphi_{2}^f(x_{(1)}),\dots,\upphi_{L}^f(x_{(L-1)}))\nonumber\\
        &=-\sum_{h=1}^L\frac{\partial\Psi_{h}}{x_{(1)}}(x_{(h)})\left[-\xh+\sum_{k=1}^{L}W_{hk}\Psi_{k}(x_{(k)})\right]\nonumber\\
        &=f_{E^1}(x_{(1)}),
    \end{align}
    where the interaction matrices $\{W_{hk}\in\real^{N_h\times N_{k}}\}_{h,k=1}^L$ are defined in~\eqref{eq: m-W}.
\end{definition}
\begin{remark}[Implementability of hybrid EBM architecture]
The formulation in Definition~\ref{def: h-EBM} leads to a highly convenient implementation: once the convex primitives $\{\lagr_h\}_h$, interaction matrices $\{W_{hk}\}_{hk}$, and biases $\{b_h\}_h$ are specified, the EBM is fully determined by the energy function $\En(x)$, and the neural network action is given by $-\nabla_{x_{(1)}} \En(x)$.
\end{remark}
Theorem~\ref{thm: fwd-inv} characterizes sufficient conditions under which the hybrid visible-layer dynamics admit a compact absorbing invariant set. The assumptions reflect typical architectural choices: linear visible activation, bounded hidden activations, and feedforward connections consistent with the EBM structure.
\begin{theorem}[Absorbing invariance of hybrid EBMs]\label{thm: fwd-inv}
    Let the hybrid EBM from Definition~\ref{def: h-EBM} satisfy:
    \begin{itemize}
        \item[(i)] \textbf{linear visible layer}: $\Psi_{1}(x)=x$;
        \item[(ii)] \textbf{bounded first hidden layer}: $\Psi_{2}:\real^{N_2}\to K\ssubset \real^{N_2}$;
        \item[(iii)] \textbf{feedforward hidden maps}: for $h=2,\dots, L$
        \begin{equation}
            \upphi_{h}^f(x_{(h-1)})=W_{h(h-1)}\Psi_{h-1}(x_{(h-1)}).
        \end{equation}
    \end{itemize}
    Then there exists $r>0$ such that the ball $B_{r}(0)\subset\real^{N}$ is $f_{E^1}$-absorbing invariant.
\end{theorem}
\begin{proof}
    Define the \textbf{feedback hidden maps}
    \begin{equation}\label{eq: feedb}
        \upphi_h^b(x_{(h+1)})=W_{h(h+1)}\Psi_{(h+1)}(x_{(h+1)}).
    \end{equation}
    Using assumption (iii), the dissipation term $-\xh$ cancels the feedforward term for every hidden layer. Together with (i), the visible dynamics~\eqref{eq: visible-d} become
    \begin{equation}\label{eq: EBM-h-red}
        f_{E^1}(x_{(1)})=-x_{(1)}+\sum_{h=1}^{L-1} \frac{\partial \Psi_{h}}{\partial x_{(1)}}(\xh)\upphi_h^b(x_{(h+1)}),
    \end{equation}
    where we have used $\partial\Psi_1/\partial x_{(1)}\equiv \mathcal{I}_{N_1}$ for the dissipation term of the visible layer.

    From hypothesis (ii) notice that $\Psi_{2}(\upphi_{2}^f(x_{(1)}))\in K$ for all $x_{(1)}\in\real^{N_{1}}$, and consequently
    \begin{equation}
        \upphi_3^f(x_{(2)})\in K_3\ssubset\real^{N_{3}}.
    \end{equation}
    Since $\{\upphi^f_h\}_{h=2}^L$ are all continuous functions
    \begin{itemize}
        \item \textbf{Boundedness of the feedforward maps $\upphi^f_h$}:\\
        For all $h=3,\dots,L$ there exists $K_{h}\ssubset \real^{N_{h}}$ such that the feedforward map
        \begin{align}
            \xh&=\upphi_h^f(x_{(h-1)})\nonumber\\
            &=\upphi_h^f\circ\dots\circ\upphi_2^f(x_{(1)})\in K_{h},
        \end{align}
        for all $x_{(1)}\in\real^{N_1}$.
        \item \textbf{Boundedness of the feedback maps $\upphi_h^b$}:\\
        By definition~\eqref{eq: feedb}, $\{\upphi_h^b\}_{h=1}^{L-1}$ are also continuous. Exploiting the boundedness of the feedforward maps, for all $h=2,\dots,L-1$ there exist $\Tilde K_{h}\ssubset \real^{N_{h}}$ such that the feedback map
        \begin{align}
            \upphi_h^b(x_{(h+1)})&=\upphi^b_h(\upphi^f_{h+1}(\xh))\nonumber\\
            &=\upphi^b_h(\upphi^f_{h+1}\circ\dots\circ\upphi_2^f(x_{(1)}))\in \Tilde K_{h}.
        \end{align}
        Additionally, observe that $\upphi_1^b(x_{(2)})=W_{12}\Psi_2(x_{(2)})$, and from hypothesis (ii) we recover boundedness of the feedback maps for all $h=1,\dots,L-1$.
        \item \textbf{Boundedness of the Jacobians $\partial\Psi_h/\partial x_{(1)}$}:\\
        Since $\Psi_2:\real^{N_2}\to K$ and $K\ssubset \real^{N_2}$ is compact, then necessarily $\|\partial\Psi_2/\partial x_{(1)}\|<M$ uniformly for some $M>0$. Since $\Psi_h$ are locally Lipschitz, then $\partial\Psi_h/\partial x_{(1)}$ are bounded on every compact in $\real^{N_{h}}$ for $h\geq 3$.
    \end{itemize}
    Express $\xh=\upphi^f_h\circ\dots\circ\upphi_2^f(x_{(1)})$ and let
    \begin{equation}
        \upgamma_{h}(x_{(1)})
        =\left\|\frac{\partial \Psi_h}{\partial x_{(1)}}(\xh)\upphi_h^b\left(\upphi_{h+1}^f(\xh)\right)\right\|_{2}.
    \end{equation}
    For all $h=1,\dots,L-1$ there exists
    \begin{align}
        \upgamma_h &= \sup_{x_{(1)}\in\real^{N_1}} \upgamma_h(x_{(1)})\nonumber\\
        &=\max_{\xh\in K_h} \left\|\frac{\partial \Psi_h}{\partial x_{(1)}}(\xh)\upphi_h^b\left(\upphi_{h+1}^f(\xh)\right)\right\|_{2}.
    \end{align}
    Define the radius
    \begin{equation}\label{eq: rad}
        r = {\textstyle\sum_{h=1}^{L-1}}\upgamma_h.
    \end{equation}
    and the ball $B_{r}(0)\in\real^{N_{1}}$, for any $x_{(1)}\in\partial B_{r}(0)$ ($x_{(1)}/r$ outer normal)
    \begin{align}
        &x_{(1)}^\top f_{E^1}(x_{(1)})=\nonumber\\
        &=-\|x_{(1)}\|_2^2+x_{(1)}^{\top}{\textstyle\sum_{h=1}^{L-1}} \frac{\partial \Psi_{h}}{\partial x_{(1)}}(\xh)\upphi^b_h(\upphi_{h+1}^f(\xh))\nonumber\\
        &\leq -r^2+r\left\|{\textstyle\sum_{h=1}^{L-1}} \frac{\partial \Psi_{h}}{\partial x_{(1)}}(\xh)\upphi^b_h(\upphi_{h+1}^f(\xh))\right\|_2\nonumber\\
        &\leq -r^2 +r \left({\textstyle\sum_{h=1}^{L-1}}\upgamma_h(x_{(1)})\right)\nonumber\\
        &\leq -r^2+r\left({\textstyle\sum_{h=1}^{L-1}}\upgamma_h\right)=-r^2+r^2=0.
    \end{align}
    Finally, observe that for any $R>r$ we have $x_{(1)}^\top f_{E^1}(x_{(1)})<0$ for all $x_{(1)}\in\partial B_{R}(0)$, and by continuity of $f_{E^1}$ there exists a $\upgamma_R>0$ such that $x_{(1)}^\top f_{E^1}(x_{(1)})<-\upgamma_R$ uniformly in $\partial B_{R}(0)$. Therefore, for all $\updelta t>0$ there exists a $0<\varepsilon<R-r$ such that if $x_0\in\partial B_{R}(0)$, then $\Phi_{f_{E^1}}^{\updelta t}(x_0)\in B_{R-\varepsilon}(0)$. Consequently, $\lim_{t\to +\infty} \Phi_{f_{E^1}}^{\updelta t}(x)\in B_r(0)$ for all $x\in B_R(0)$, and taking $R\to +\infty$, we conclude.
\end{proof}
Theorem~\ref{thm: fwd-inv} weakens the conditions required in Proposition~\ref{prop: EBM-ru} and enables the use of expressive architectures with arbitrary deep hidden-layer activations while preserving stability, thereby fulfilling the original aim of~\citep{KD-HJJ:16} to broaden the admissible activation class for EBMs. Corollary~\ref{cor: PHD} below establishes that absorbing invariance and dissipativity of the projected visible-layer EBM dynamics~\eqref{eq: l-EBM-1} are robust under a Port‑Hamiltonian transformations, ensuring that the stability guarantees derived for EBMs carry over to a larger and practically relevant class of dynamics.
\begin{corollary}[Absorbing invariance and dissipativity of Port-Hamiltonian EBMs]\label{cor: PHD}
    Let $Q:\real^{N_1}\to\real^{{N_1}\times {N_1}}$ be continuous and uniformly bounded, and assume $Q(x)+Q(x)^\top\succ 0$ for all $x\in\real^{N_1}$. Let $f_{E^1}(x)$ denote the EBM vector field~\eqref{eq: EBM-h-red}, admitting an absorbing invariant ball $B_r(0)$ with radius $r>0$ defined in~\eqref{eq: rad}. Then there exist $q_{\max}>q_{\min}>0$ and $\uprho:=\frac{q_{\max}}{q_{\min}}r$ such that the Port-Hamiltonian EBM vector field
    \begin{equation}
        \mathcal{H}_E(x) = Q(x)f_{E^1}(x),
    \end{equation}
    is absorbing invariant on every ball $B_{R}(0)$ for all $R\geq \uprho$. Furthermore,
    \begin{equation}
        \mathcal{L}_{\mathcal{H}_E}\En(x)\leq 0\qquad \text{a.e.}
    \end{equation}
\end{corollary}
\begin{proof}
    Define
    \begin{equation}
        q_{\text{min}} = \inf_{x\in\real^{N_1}/{\vectorzeros[N_1]}} \frac{x^\top Q(x) x}{x^\top x};\quad q_{\max} = \max_{x\in\real^{N_1}} \|Q(x)\|_2,\nonumber
    \end{equation}
    with $q_{\min}>0$ by uniform positive definiteness and $0<q_{\max}<+\infty$ by boundedness.
    Then for $x\in \partial B_{R}(0)$ with $R>r$ we have
    \begin{align}
        x^{\top}\mathcal{H}_E(x)&=x^{\top}Q(x)f_{E^1}(x)\nonumber\\
        &=x^{\top}Q(x)\left[-x+\textstyle\sum_{h=1}^{L-1} \frac{\partial \Psi_{h}}{\partial x}(\xh)\upphi_h^b(x_{(h+1)})\right]\nonumber\\
        &\leq -q_{\min}\|x\|_2^2+q_{\max}r\|x\|_2\nonumber\\
        &=-R(q_{\min}R-q_{\max}r).
    \end{align}
    The inner product is negative for any 
    \begin{equation}\label{eq: rad-Q}
        R>\uprho:=\frac{q_{\max}}{q_{\min}}r,
    \end{equation}
    and following the same approach as in~\ref{thm: fwd-inv}, we conclude on the absorbing invariance. Finally, observe that at point of differentiability $\nabla \En(x)=f_{E^1}(x)$~\eqref{eq: EBM-h-red}, and consequently the Clarke Lie derivative of $\En$ along $\mathcal{H}_E$ is
    \begin{align}
        \mathcal{L}_{\mathcal{H}_E}\En(x) &= \max_{v\in\partial\En(x)} v^\top \mathcal{H}_{E}(x)\nonumber\\
        &=\max_{v\in\partial\En(x)} -v^\top Q(x) v\leq 0,
    \end{align}
    by the positive definiteness of $Q(x)+Q(x)^\top$.
\end{proof}
We have proven how absorbing invariance and dissipativity results for hybrid EBMs extends seamlessly to a broader class of Port‑Hamiltonian~\citep{VDSA:07} energy‑based systems, where the dissipative gradient flow is composed with a state‑dependent metric. In this setting, the EBM vector field is premultiplied by a continuous matrix field $Q(x)$ whose symmetric part remains uniformly positive definite, thereby preserving dissipativity while introducing geometry‑dependent rotations. Port-Hamiltonian models capture dynamics evolving in non‑Euclidean geometries, as well as systems subjected to rotational or gyroscopic perturbations, and play a central role in geometric control, robotics, and structure‑preserving numerical schemes. Therefore, our proposed Port-Hamiltonian hybrid EBM architecture naturally provides a safe and \emph{ready-to-use} data-driven paradigm for system identification in many relevant and safety-critical real-world settings.

\section{MODEL VALIDATION}
\noindent\textbf{Setting.} The hybrid EBM validation focuses on identifying nonlinear dynamics generated by the gradient of nontrivial potentials pre-multiplied by a state-dependent metric. To showcase the ability of EBMs to identify rich, strongly nonlinear dynamics we benchmark the proposed framework on two characteristic systems on $\mathbb{R}^2$:
\begin{enumerate}
	\item a \textbf{multi-well potential} $V_{m}:\real^2\to\real$ exhibiting multiple convex-concave regions
	\begin{equation}\label{eq: Vm}
		V_{m}(x,y)=\upalpha(x^2+\upbeta y^2)+\upomega(\sin(\upgamma x)\cos(y)),
	\end{equation}
	\item an \textbf{exotic potential} $V_{e}:\real^2\to\real$ characterized by multiple local minima and a ring-shaped invariant set where the vector field vanishes
	\begin{align}
		V_{e}(x,y)=&\upalpha(x^2+y^2-1)^2+\upomega(\sin(\upgamma x)\cos(\upzeta y))
    +\upxi\sin(x^2-y^2)+\uptheta\cos(\uprho xy)
    +\updelta(x^2y-y^3).
	\end{align}
\end{enumerate}
where the parameters $\upalpha,\upbeta,\upomega,\upgamma,\upzeta,\upxi,\uptheta,\updelta>0$ are arbitrarily chosen to emphasize specific qualities of the dynamics, i.e., ruggedness of the potential vs induced rotation. Both potentials are evaluated under a non-Euclidean metric deformation, implemented through a smooth $\sin$-$\cos$ dependent matrix
\begin{equation}
	Q(x)=
	\begin{pmatrix}
		1+\sin(x)/2 & 0.3\cos(y)\\
		0.3\cos(y) & 1+\cos(x+y)/2
	\end{pmatrix},
\end{equation}
leading to ground‑truth dynamics $f_m(x,y)=Q(x,y)\nabla V_m(x,y)$ and $f_e(x,y)=Q(x,y)\nabla V_e(x,y)$. 

\noindent\textbf{Methods.} \textit{Data generation:}
We sample $2000$ initial conditions in $[-2,2]^2$ and generate trajectories of length $T=10$ using Euler integration with $\delta t=0.01$. States and vector-field evaluations $((x_t,y_t), f(x_t,y_t))$ form the dataset, which is split into $1600$ training and $400$ test trajectories. \textit{Architectures:} For the multi-well system we use a hybrid EBM with one softmax hidden layer ($M=128$). For the exotic system we adopt a two-layer architecture (softmax layer with $M=512$ and polynomial layer with $M=128$). In all cases, the visible layer is dynamical, hidden layers are feedforward, and the metric is learned via an MLP. \textit{Training:} The model is trained with a composite loss combining weighted field error, short-horizon rollout error, and $\ell_2$ regularization. Optimization uses AdamW, and hyperparameters are selected through Weights \& Biases (wandb) sweeps\footnote{Code avaliable at: \url{https://github.com/sim1bet/EBM_StableSystemIdentification}}.

\noindent\textbf{Numerical results.} The hybrid EBM learns a Port‑Hamiltonian representation $\mathcal{H}_E(x,y)$ of these fields. Table~\ref{tab: loss} reports the final train and test losses after $5000$ learning epochs for the identification of the two systems. In both cases, the hybrid EBM achieves low reconstruction error, with test losses remaining in the same range as the training error despite the presence of nonlinear metric distortions.
\begin{table}
    \begin{center}
        \begin{tabular}{|c||c||c|}
        Dynamics & Train Loss & Test Loss\\
        \hline
        $f_{m}$ & $0.011$ & $0.002$\\
        \hline
        $f_{e}$ & $0.047$ & $0.042$\\
        \hline
    \end{tabular}
    \end{center}
    \caption{Summary of train and test loss at the end of $n=5000$ learning epochs for system identification of $f_m$ and $f_e$ using a Port-Hamiltonian EBM $\mathcal{H}_E$.}
    \label{tab: loss}
\end{table}
\begin{figure}[!h]
\centering
\begin{minipage}{0.48\linewidth}
    \centering
    \includegraphics[width=\linewidth]{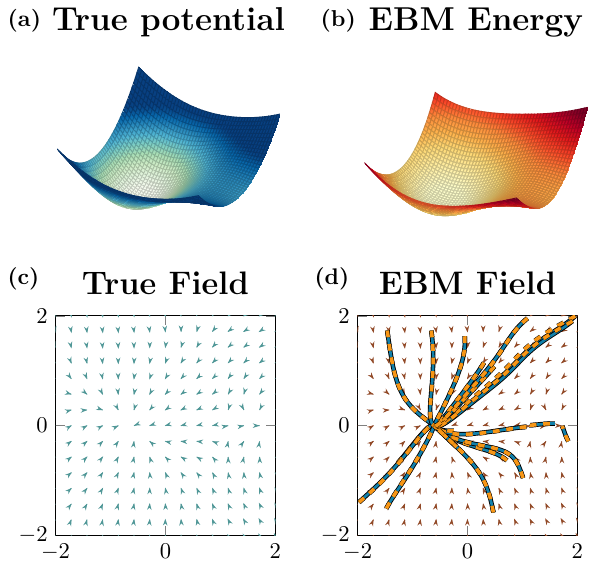}
    \caption{\textbf{EBM identification of multi-well potential dynamics}. (a) Ground‑truth potential $V_m(x,y)$. (b) EBM energy $\En(x,y)$ accurately recovering the landscape geometry. (c) True vector field. (d) Reconstructed field with true and EBM trajectories, showing consistent flow patterns and convergence to the correct attractors.\\}
    \label{fig: saddle_res}
\end{minipage}
\hfill
\begin{minipage}{0.48\linewidth}
    \centering
    \includegraphics[width=\linewidth]{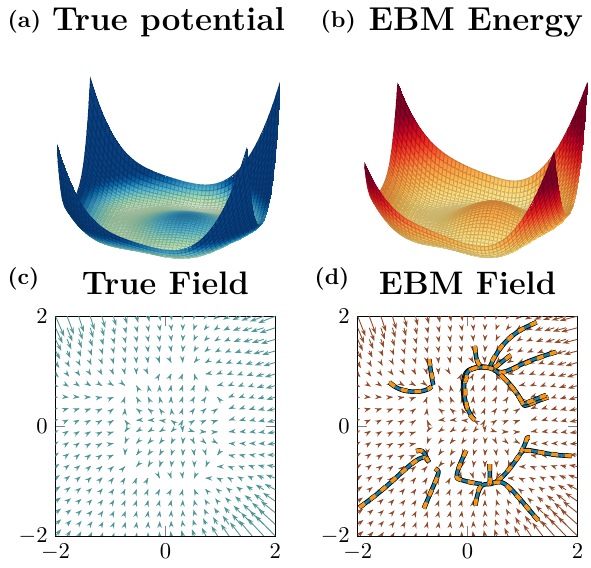}
    \caption{\textbf{EBM identification of exotic potential dynamics}. (a) Ground‑truth potential $V_e(x,y)$. (b) EBM energy $\En(x,y)$ accurately recovering the landscape ring-geometry. (c) True vector field. (d) Reconstructed field with true and EBM trajectories. EBM dynamics are consistent with the true vector field and reliably recover non-linear, rotational components.}
    \label{fig: ring_res}
\end{minipage}

\end{figure}
For the multi-well system $f_m$, the learned energy $\En$ reproduces the geometry of the true potential $V_m$, as shown in Figures~\ref{fig: saddle_res}(a)–(b). Despite a different 
magnitude scaling, the convex–concave structure is accurately captured. The reconstructed vector field and trajectories (Figures~\ref{fig: saddle_res}(c)–(d)) closely match the true 
dynamics: quiver fields are visually indistinguishable, and trajectories initialized at the same state converge to the same minima, confirming the model’s ability to recover both local flow geometry and global attractor structure.

For the exotic system $f_e$, which features a ring-shaped zero-field region and strong metric-induced distortions, the hybrid EBM accurately recovers both the global basin structure and the central ring (Figures~\ref{fig: ring_res}(a)–(b)). The reconstructed vector field and trajectories (Figures~\ref{fig: ring_res}(c)–(d)) remain consistent with the true dynamics, capturing rotational components and preserving the qualitative behavior predicted by the theoretical stability analysis. These results demonstrate that the proposed Port-Hamiltonian hybrid EBM reliably identifies nonlinear flows on non-Euclidean geometries while maintaining structural stability and interpretability through its energy-based formulation.

Direct computation of the invariance radius $r$ in~\eqref{eq: rad-Q} is generally intractable for practical EBM architectures: the width of the first hidden layer determines the dimension of the activation function image, and the corresponding search space grows exponentially. In our setting, the smallest first hidden layer has $M=128$ units, and its activation co-domain is $[0,1]^M$, making an exhaustive evaluation of the quantities entering~\eqref{eq: rad-Q} computationally prohibitive. To obtain a meaningful and computationally feasible estimate, we instead compute an \emph{expansion radius} $r_{\mathrm{ex}} < r$ using only the data-driven range explored during training. Specifically, since the dataset $\mathcal{X}$ is generated synthetically over the mesh $\mathcal{A} = [-2,2]^2$, we forward-propagate all points in $\mathcal{A}$ through the first hidden layer and use the restricted activation image $\upphi_2^f(\mathcal{A})$ to evaluate the maximal value of the terms entering~\eqref{eq: rad-Q}. This leads to a conservative but sound and computationally tractable estimate $r_{\mathrm{ex}}$ that suffices to certify absorbing invariance over the domain explored by the learned dynamics. As illustrated in Figure~\ref{fig: rad}, the expansion radius $r_{\mathrm{ex}}$ is largely insensitive to network depth (one hidden layer for $V_m$ versus two for $V_e$) prior to metric deformation. Following multiplication by the factor $q_{\max}/q_{\min}$, the radius increases by approximately one order of magnitude.
\begin{figure}[!h]
	\centering
	\includegraphics[width=.75\linewidth]{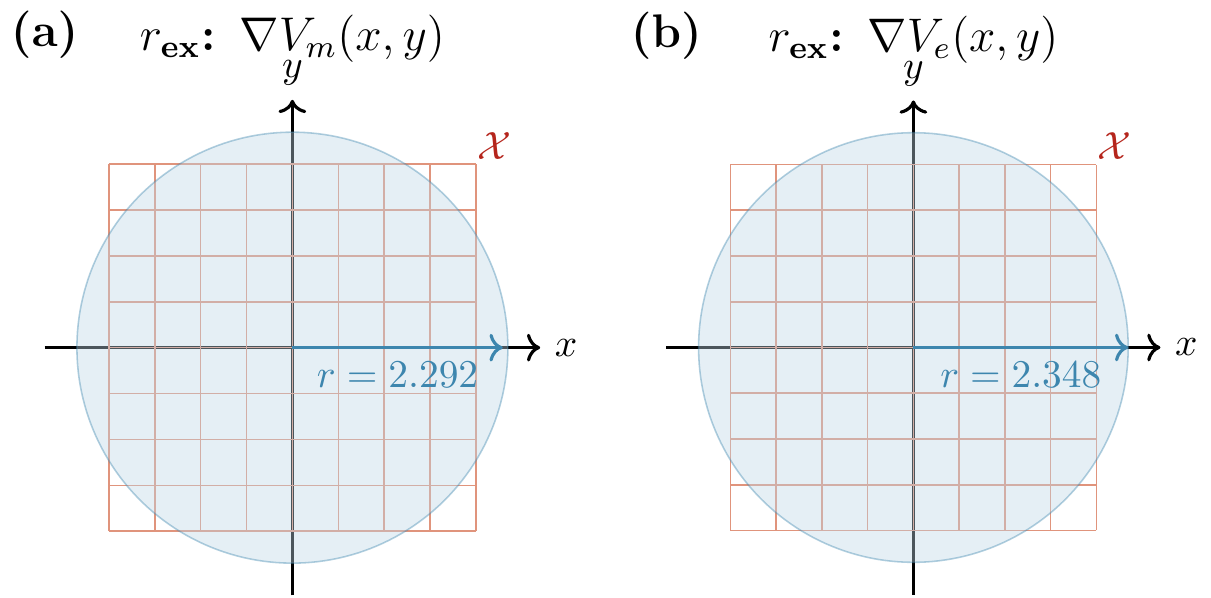}
	\caption{\textbf{Expansiveness radius}. Maximal radius of expansiveness given the data $\mathcal{X}$. The radii have similar magnitude for both the (a) multi-well potential $V_{m}(x,y)$ and the (b) exotic potential $V_e(x,y)$.}
	\label{fig: rad}
\end{figure}

\section{CONCLUSIONS}
\emph{Energy-based models} bridge dynamical systems theory and neural networks, offering expressive models that are \emph{stable-by-design}. This work advances their foundations along three dimensions. First, we generalize classical EBM stability analysis to locally Lipschitz activations, establishing energy dissipation via Clarke derivatives and deriving conditions for radial unboundedness beyond the $C^2$ setting. Second, we introduce a hybrid EBM architecture with dynamic visible layer and static feedforward-feedback hidden layer maps. We prove absorbing invariance of the hybrid EBM dynamics under mild assumptions and extend this guarantee to port-Hamiltonian flows, enabling the modeling of systems on non-Euclidean geometries. Finally, we validate the framework on metric-deformed multi-well and ring-shaped potentials, demonstrating accurate vector-field reconstruction, faithful recovery of energy-landscape geometry, and reliable trajectory tracking. These results highlight hybrid EBMs as physically grounded, stable-by-design models suited for data-driven identification in nonlinear, safety critical domains. Future work will explore integrating control inputs to enable safe and certifiable identification, further advancing energy-based modeling as a core methodology for physical AI.

\section{ACKNOWLEDGMENTS}

The authors gratefully acknowledge the contribution of IT4LIA AI Factory project EHPC-AIF-2026PG01-103.

\newpage
\bibliographystyle{unsrtnat}
\bibliography{SB-main}

@article{AKJ-EP:71,
  title = {System identification—A survey},
  volume = {7},
  ISSN = {0005-1098},
  DOI = {10.1016/0005-1098(71)90059-8},
  number = {2},
  journal = {Automatica},
  publisher = {Elsevier BV},
  author = {Åstr\"{o}m,  K.~J. and Eykhoff,  P.},
  year = {1971},
  month = mar,
  pages = {123–162}
}

@article{AL:24,
  title = {In Search of Dispersed Memories: Generative Diffusion Models Are Associative Memory Networks},
  volume = {26},
  ISSN = {1099-4300},
  DOI = {10.3390/e26050381},
  number = {5},
  journal = {Entropy},
  publisher = {MDPI AG},
  author = {Ambrogioni,  L.},
  year = {2024},
  month = apr,
  pages = {381}
}

@article{BL-CFE-NJ:18,
  title = {Optimization Methods for Large-Scale Machine Learning},
  volume = {60},
  ISSN = {1095-7200},
  DOI = {10.1137/16m1080173},
  number = {2},
  journal = {SIAM Review},
  publisher = {Society for Industrial & Applied Mathematics (SIAM)},
  author = {Bottou,  L. and Curtis,  F.~E. and Nocedal,  J.},
  year = {2018},
  month = jan,
  pages = {223–311}
}

@book{BHH-CPL:17,
  title = {Convex Analysis and Monotone Operator Theory in Hilbert Spaces},
  ISBN = {9783319483115},
  ISSN = {2197-4152},
  DOI = {10.1007/978-3-319-48311-5},
  journal = {CMS Books in Mathematics},
  publisher = {Springer International Publishing},
  author = {Bauschke,  H.~H. and Combettes,  P.~L.},
  year = {2017}
}

@article{CFH:75,
  title = {Generalized gradients and applications},
  volume = {205},
  ISSN = {0002-9947},
  DOI = {10.1090/s0002-9947-1975-0367131-6},
  journal = {Transactions of the American Mathematical Society},
  publisher = {American Mathematical Society (AMS)},
  author = {Clarke,  F.~H.},
  year = {1975},
  pages = {247–247}
}

@article{CMA-GS:83,
  title = {Absolute stability of global pattern formation and parallel memory storage by competitive neural networks},
  volume = {SMC-13},
  ISSN = {2168-2909},
  DOI = {10.1109/tsmc.1983.6313075},
  number = {5},
  journal = {IEEE Transactions on Systems,  Man,  and Cybernetics},
  publisher = {Institute of Electrical and Electronics Engineers (IEEE)},
  author = {Cohen,  M.~A. and Grossberg,  S.},
  year = {1983},
  month = sep,
  pages = {815–826}
}

@inproceedings{CRTQ-RY-BJ:18,
  author = {Chen, R.~T.~Q. and Rubanova, Y. and Bettencourt, J. and Duvenaud, D.~K.},
  doi = {10.48550/arXiv.1806.07366},
  booktitle = {Advances in Neural Information Processing Systems},
  editor = {S. Bengio and H. Wallach and H. Larochelle and K. Grauman and N. Cesa-Bianchi and R. Garnett},
  pages = {},
  publisher = {Curran Associates, Inc.},
  title = {Neural Ordinary Differential Equations},
  volume = {31},
  year = {2018}
}

@article{CR-JH-SGS:22,
  title = {A review of active probing-based system identification techniques with applications in power systems},
  volume = {140},
  ISSN = {0142-0615},
  DOI = {10.1016/j.ijepes.2022.108008},
  journal = {International Journal of Electrical Power \& Energy Systems},
  publisher = {Elsevier BV},
  author = {Chakraborty,  R. and Jain,  H. and Seo,  G.~S.},
  year = {2022},
  month = sep,
  pages = {108008}
}

@article{DM-HJ-LM:17,
  title =	 {On a Model of Associative Memory with Huge Storage
                  Capacity},
  volume =	 168,
  ISSN =	 {1572-9613},
  DOI =      {10.1007/s10955-017-1806-y},
  number =	 2,
  journal =	 {Journal of Statistical Physics},
  publisher =	 {Springer Science and Business Media LLC},
  author =	 {Demircigil, M. and Heusel, J. and Lowe, M. and Upgang,
                  S. and and Vermet, F.},
  year =	 2017,
  month =	 may,
  pages =	 {288–299}
}

@article{DP-DC-HJA:18,
  title = {Dynamic production system identification for smart manufacturing systems},
  volume = {48},
  ISSN = {0278-6125},
  DOI = {10.1016/j.jmsy.2018.04.006},
  journal = {Journal of Manufacturing Systems},
  publisher = {Elsevier BV},
  author = {Denno,  P. and Dickerson,  C. and Harding,  J.~A.},
  year = {2018},
  month = jul,
  pages = {192–203}
}

@article{DH-LB-YL:21,
  doi = {10.48550/ARXIV.2109.14152},
  author = {Dai,  H. and Landry,  B. and Yang,  L. and Pavone,  M. and Tedrake,  R.},
  title = {Lyapunov-stable neural-network control},
  journal = {arXiv},
  year = {2021},
  copyright = {Creative Commons Attribution 4.0 International}
}

@article{DJ-TA-VS:22,
  title = {Dissipative Deep Neural Dynamical Systems},
  volume = {1},
  ISSN = {2694-085X},
  DOI = {10.1109/ojcsys.2022.3186838},
  journal = {IEEE Open Journal of Control Systems},
  publisher = {Institute of Electrical and Electronics Engineers (IEEE)},
  author = {Drgona,  J. and Tuor,  A. and Vasisht,  S. and Vrabie,  D.},
  year = {2022},
  pages = {100–112}
}

@book{ELC-GRF:15,
  title = {Measure Theory and Fine Properties of Functions,  Revised Edition},
  ISBN = {9781482242393},
  DOI = {10.1201/b18333},
  publisher = {Chapman and Hall/CRC},
  author = {Evans,  L.~C. and Gariepy,  R.~F.},
  year = {2015},
  month = apr 
}

@article{HJJ:82,
  title = {Neural networks and physical systems with emergent collective computational abilities.},
  volume = {79},
  ISSN = {1091-6490},
  DOI = {10.1073/pnas.79.8.2554},
  number = {8},
  journal = {Proceedings of the National Academy of Sciences},
  publisher = {Proceedings of the National Academy of Sciences},
  author = {Hopfield,  J.~J.},
  year = {1982},
  month = apr,
  pages = {2554–2558}
}

@article{HJJ:84,
  title = {Neurons with graded response have collective computational properties like those of two-state neurons.},
  volume = {81},
  ISSN = {1091-6490},
  DOI = {10.1073/pnas.81.10.3088},
  number = {10},
  journal = {Proceedings of the National Academy of Sciences},
  publisher = {Proceedings of the National Academy of Sciences},
  author = {Hopfield,  J.~J.},
  year = {1984},
  month = may,
  pages = {3088–3092}
}

@inproceedings{HB-CDH-SH:22,
  title={A Universal Abstraction for Hierarchical Hopfield Networks},
  author={Hoover, B. and Chau, D.~H. and Strobelt, H. and Krotov, D.},
  booktitle={The Symbiosis of Deep Learning and Differential Equations II},
  year={2022},
}

@article{HB-LY-PB:23,
  title={Energy transformer},
  doi={10.48550/arXiv.2302.07253},
  author={Hoover, B. and Liang, Y. and Pham, B. and Panda, R. and Strobelt, H. and Chau, D.~H. and Zaki, M. and Krotov, D.},
  journal={Advances in neural information processing systems},
  volume={36},
  pages={27532--27559},
  year={2023}
}

@inproceedings{HB-CDH-SH:24,
  author = {Hoover, B. and Chau, D.~H. and Strobelt, H. and Ram, P. and Krotov, D.},
  booktitle = {Advances in Neural Information Processing Systems},
  doi = {10.52202/079017-0742},
  editor = {A. Globerson and L. Mackey and D. Belgrave and A. Fan and U. Paquet and J. Tomczak and C. Zhang},
  pages = {23549--23576},
  publisher = {Curran Associates, Inc.},
  title = {Dense Associative Memory Through the Lens of Random Features},
  volume = {37},
  year = {2024}
}

@inproceedings{KD-HJJ:16,
  title =	 {Dense associative memory for pattern recognition},
  booktitle =	 {Advances in neural information processing systems},
  author =	 {Krotov, D. and Hopfield, J.~J.},
  volume =	 29,
  year =	 2016,
  DOI =		 {10.48550/arXiv.1606.01164},
}

@article{KL-SJJ-KD:25,
  title = {Neuron–astrocyte associative memory},
  volume = {122},
  ISSN = {1091-6490},
  DOI = {10.1073/pnas.2417788122},
  number = {21},
  journal = {Proceedings of the National Academy of Sciences},
  publisher = {Proceedings of the National Academy of Sciences},
  author = {Kozachkov,  L. and Slotine,  J.~J. and Krotov,  D.},
  year = {2025},
  month = may 
}

@inproceedings{KD-HJJ:20,
  title =	 {Large Associative Memory Problem in Neurobiology and
                  Machine Learning},
  booktitle =	 {International Conference on Learning Representations},
  author =	 {Krotov, D. and Hopfield, J.~J.},
  year =	 2020,
  DOI = {10.48550/arXiv.2008.06996},
}

@article{LL:10,
  title = {Perspectives on system identification},
  volume = {34},
  ISSN = {1367-5788},
  DOI = {10.1016/j.arcontrol.2009.12.001},
  number = {1},
  journal = {Annual Reviews in Control},
  publisher = {Elsevier BV},
  author = {Ljung,  L.},
  year = {2010},
  month = apr,
  pages = {1–12}
}

@book{LJM:12,
  title = {Introduction to Smooth Manifolds},
  ISBN = {9781441999825},
  ISSN = {0072-5285},
  DOI = {10.1007/978-1-4419-9982-5},
  journal = {Graduate Texts in Mathematics},
  publisher = {Springer New York},
  author = {Lee,  J.~M.},
  year = {2012}
}

@article{MR-PE-RE:87,
  title = {The capacity of the Hopfield associative memory},
  volume = {33},
  ISSN = {1557-9654},
  DOI = {10.1109/tit.1987.1057328},
  number = {4},
  journal = {IEEE Transactions on Information Theory},
  publisher = {Institute of Electrical and Electronics Engineers (IEEE)},
  author = {McEliece,  R. and Posner,  E. and Rodemich,  E. and Venkatesh,  S.},
  year = {1987},
  month = jul,
  pages = {461–482}
}

@article{MS-PM-BM:20,
  doi = {10.48550/ARXIV.2003.08063},
  author = {Massaroli,  S. and Poli,  M. and Bin,  M. and Park,  J. and Yamashita,  A. and Asama,  H.},
  title = {Stable Neural Flows},
  journal = {arXiv},
  year = {2020},
  copyright = {arXiv.org perpetual,  non-exclusive license}
}

@article{PA-GR-YW:22,
  title = {Stable robot manipulator parameter identification: A closed-loop input error approach},
  volume = {141},
  ISSN = {0005-1098},
  DOI = {10.1016/j.automatica.2022.110294},
  journal = {Automatica},
  publisher = {Elsevier BV},
  author = {Perrusquía,  A. and Garrido,  R. and Yu,  W.},
  year = {2022},
  month = jul,
  pages = {110294}
}

@inproceedings{PGY-KJ-KB:23,
  author = {Park, G.~Y. and Kim, J. and Kim, B. and Lee, S.~W. and Ye, J.~C.},
  doi = {doi.org/10.48550/arXiv.2306.09869},
  booktitle = {Advances in Neural Information Processing Systems},
  editor = {A. Oh and T. Naumann and A. Globerson and K. Saenko and M. Hardt and S. Levine},
  pages = {76382--76408},
  publisher = {Curran Associates, Inc.},
  title = {Energy-Based Cross Attention for Bayesian Context Update in Text-to-Image Diffusion Models},
  volume = {36},
  year = {2023}
}

@inproceedings{RH-SB-LJ:21,
  title={Hopfield Networks is All You Need},
  doi={10.48550/arXiv.2008.02217},
  author={Ramsauer, H. and Sch{\"a}fl, B. and Lehner, J. and Seidl, P. and Widrich, M. and Gruber, L. and Holzleitner, M. and Adler, T. and Kreil, D. and Kopp, M.~K. and Klambauer, G. and Brandstetter, J. and Hochreiter, S.},
  year={2021},
  booktitle={International Conference on Learning Representations}
}

@article{RFJ-KDK-KM:25,
  doi = {10.48550/ARXIV.2502.02480},
  author = {Roth,  F.~J. and Klein,  D.~K. and Kannapinn,  M. and Peters,  J. and Weeger,  O.},
  title = {Stable Port-Hamiltonian Neural Networks},
  journal = {arXiv},
  year = {2025},
  copyright = {Creative Commons Attribution 4.0 International}
}

@book{VDSA:07,
  title = {Port-Hamiltonian systems: an introductory survey},
  ISBN = {9783985475384},
  DOI = {10.4171/022-3/65},
  booktitle = {Proceedings of the International Congress of Mathematicians Madrid,  August 22–30,  2006},
  publisher = {EMS Press},
  author = {Van der Schaft,  A.},
  year = {2007},
  month = may,
  pages = {1339–1365}
}

@inproceedings{WW-HTY-HJYC:25,
    title={In-Context Learning as Conditioned Associative Memory Retrieval},
    author={Wu, W. and Hsiao, T.~Y. and Hu, J.~Y.~C. and Zhang, W. and Liu, H.},
    booktitle={Forty-second International Conference on Machine Learning},
    year={2025},
}

@article{XY-SS:23,
  title = {Learning Dissipative Neural Dynamical Systems},
  volume = {7},
  ISSN = {2475-1456},
  DOI = {10.1109/lcsys.2023.3337851},
  journal = {IEEE Control Systems Letters},
  publisher = {Institute of Electrical and Electronics Engineers (IEEE)},
  author = {Xu,  Y. and Sivaranjani,  S.},
  year = {2023},
  pages = {3531–3536}
}

\end{document}